\documentclass[runningheads]{llncs}

\usepackage[T1]{fontenc}
\usepackage{graphicx}
\usepackage{xcolor}
\definecolor{linkblue}{rgb}{0.098,0.098,0.4392}

\usepackage[colorlinks=true,
    linkcolor=linkblue,
    anchorcolor=linkblue,
    citecolor=linkblue,
    filecolor=linkblue,
    menucolor=linkblue,
    urlcolor=linkblue,
    bookmarks=true,
    bookmarksopen=true,
    bookmarksopenlevel=2,
    bookmarksnumbered=true,
    hyperindex=true,
    plainpages=false,
    pdfpagelabels=true
]{hyperref}

\usepackage{booktabs}
\usepackage{multirow}
\usepackage{mathtools}
\usepackage{amssymb}
\usepackage[capitalize]{cleveref}
\usepackage{thm-restate}
\usepackage{threeparttable}

\graphicspath{{./figures/}}

\hypersetup{
    pdfauthor={Antonio Lauerbach, Kendra Reiter, and Marie Schmidt},
    pdftitle={The Complexity of Counting Turns in the Line-Based Dial-a-Ride Problem},
    pdfkeywords={Dial-a-Ride Problem · liDARP · NP-hardness · Parameterized Complexity}
}
\newcommand{\arxiv}[2]{#1}

\DeclarePairedDelimiter{\abs}{\lvert}{\rvert}
\DeclarePairedDelimiter{\ceils}{\lceil}{\rceil}

\DeclarePairedDelimiter{\angles}{\langle}{\rangle}
\DeclarePairedDelimiter{\set}{\{}{\}}
\DeclarePairedDelimiter{\parens}{(}{)}

\newcommand{\N}{\mathbb{N}}
\newcommand{\R}{\mathbb{R}}

\newcommand{\NP}{\textsf{NP}}
\newcommand{\XP}{\textsf{XP}}
\newcommand{\FPT}{\textsf{FPT}}

\newcommand{\optturns}{\ensuremath{\tau}}

\renewcommand{\orcidID}[1]{\href{https://orcid.org/#1}{\hspace{1pt}\raisebox{-.1\height}{\includegraphics[scale=.03]{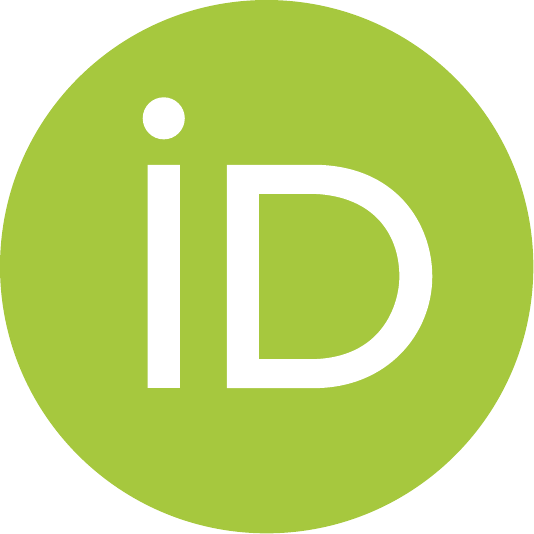}}}}

\newif\ifinappendix
\let\oldappendix\appendix
\renewcommand{\appendix}{
  \oldappendix
  \inappendixtrue
}
\newcommand{\restateref}[1]{\ifinappendix{\hyperref[#1]{$\star$}}\else{\hyperref[#1*]{$\star$}}\fi}

\makeatletter
\newcommand*{\bigcdot}{}
\DeclareRobustCommand*{\bigcdot}{%
  \mathbin{\mathpalette\bigcdot@{}}%
}
\newcommand*{\bigcdot@scalefactor}{.65}
\newcommand*{\bigcdot@widthfactor}{1.15}
\newcommand*{\bigcdot@}[2]{%
  \sbox0{$#1\vcenter{}$}
  \sbox2{$#1\cdot\m@th$}%
  \hbox to \bigcdot@widthfactor\wd2{%
    \hfil
    \raise\ht0\hbox{%
      \scalebox{\bigcdot@scalefactor}{%
        \lower\ht0\hbox{$#1\bullet\m@th$}%
      }%
    }%
    \hfil
  }%
}
\makeatother

\DeclareRobustCommand{\abbrevcrefs}{%
\crefname{theorem}{Thm.}{Thms.}%
\crefname{corollary}{Cor.}{Cors.}%
}
\DeclareRobustCommand{\cshref}[1]{{\abbrevcrefs\cref{#1}}}

\begin{document}
\title{The Complexity of Counting Turns in the Line-Based Dial-a-Ride Problem}
\author{
    Antonio Lauerbach\orcidID{0009-0007-9093-3443} \and
    Kendra Reiter\orcidID{0009-0004-7281-6516} \and
    Marie Schmidt\orcidID{0000-0001-9563-9955}
}
\authorrunning{A. Lauerbach et al.}
\institute{
    Department of Computer Science, University of Würzburg
    \email{antonio.lauerbach@stud-mail.uni-wuerzburg.de}\\
    \email{\{kendra.reiter,marie.schmidt\}@uni-wuerzburg.de}
}
\maketitle
\begin{abstract}
Dial-a-Ride problems have been proposed to model the challenge to consolidate passenger transportation requests with a fleet of shared vehicles.
The \emph{\textsc{line-based\ Dial-a-Ride\ Problem} (\textsc{liDARP})} is a variant where the passengers are transported along a fixed sequence of stops, with the option of taking shortcuts.
In this paper we consider the \textsc{liDARP} with the objective function to maximize the number of transported requests. We investigate the complexity of two optimization problems: the \textsc{liDARP}, and the problem to determine the minimum number of turns needed in an optimal \textsc{liDARP} solution, called the \emph{\textsc{MinTurn} problem}.
Based on a number of instance parameters and characteristics, we are able to state the boundary between polynomially solvable and \NP-hard instances for both problems.
Furthermore, we provide parameterized algorithms that are able to solve both the \textsc{liDARP} and \textsc{MinTurn} problem.

\end{abstract}

\section{Introduction}
Ridepooling, i.e., to flexibly serve passenger transportation requests with a fleet of shared vehicles, has been recognized as a promising option to replace conventional public transport, in particular in regions with low demand density. 
In this paper, we study the complexity of the \emph{\textsc{line-based\ Dial-a-Ride\ Problem}} (\textsc{liDARP}), which combines the spatial aspects of a fixed sequence of stops (utilizing existing infrastructure) with the temporal flexibility of ridepooling. The goal is to reduce mobility-related emissions by efficiently pooling passengers with an improved service quality due to the fixed spatial structure. Here, we consider the objective to maximize the number of transported passengers.

Secondly, we study the complexity of determining the minimum number of turns per vehicle in an optimal \textsc{liDARP} solution, which we refer to as the \textsc{MinTurn} problem. This problem may be relevant in practice, e.g., when turns of autonomous vehicles need to be supervised by a (remote) operator. It can also be relevant when formulating the \textsc{liDARP} as a mixed-integer linear program, compare~\cite{reiter_line-based_2024}.

The remainder of the paper is organized as follows. 
In \cref{sec:problem}, we formally define the \textsc{liDARP} and the \textsc{MinTurn} problems.
We summarize known complexity results and further related research findings from prior work in \cref{sec:literature}, before we give an overview of the contribution of this paper in \cref{sec:contribution}. 
In the following \cref{sec:pol,sec:hard}, we are able to characterize the boundary between polynomially solvable and \NP-hard cases of \textsc{liDARP} and \textsc{MinTurn} according to instance specifics. 
In \cref{sec:parameterized-algos} we provide parameterized algorithms for the \textsc{liDARP} and \textsc{MinTurn} problem.
We conclude by discussing open problems for further research in~\cref{sec:conclusion}. 
The (full) proofs of statements marked with a \arxiv{clickable ``$\star$''}{``$\star$''} are in the \arxiv{appendix}{full version of this paper~\cite{lauerbach_turns-lidarp_2024}}.

\section{Problem Definition}\label{sec:problem}
We start by defining the \textsc{line-based\ Dial-a-Ride\ Problem} (\textsc{liDARP}) based on the definitions by Reiter~et~al.~\cite{reiter_line-based_2024}: a \emph{line}, given by a sequence~$H$ of~$h$ \emph{stops}, is operated by~$k$ \emph{vehicles}, each with \emph{capacity}~$c$, that transport some of the~$n$ \emph{passenger requests}~$P$.

The \emph{(time) distance} between distinct stops~$i,j\in H$ is given by~$t_{i,j}\in\N$. 
The vehicles may take \emph{shortcuts} by skipping stops, \emph{wait} at a stop, or \emph{turn} (i.e., change direction with respect to the sequence of stops prescribed by the line) at a stop. To turn, a vehicle needs~$t_\textrm{turn}\in\N_0$ time (the \emph{turn time}).

Each passenger submits a \emph{request}~$p\in P$ for transportation from an \emph{origin}~$o_p\in H$ to a \emph{destination}~$d_p\in H$ with~$o_p\not=d_p$. If~$o_p$ precedes~$d_p$ in the sequence of stops, we say that the request~$p$ is \emph{ascending}, otherwise it is \emph{descending}. In contrast to~\cite{reiter_line-based_2024}, we assume that each passenger submits an individual requests, i.e., we do not allow group requests and thus the \emph{passenger load} of each requests is~$1$.
A request may specify a \emph{time window}~$[e_p,l_p]$, delimited by an \emph{earliest pick-up time}~$e_p$ and a \emph{latest drop-off time}~$l_p$, during which it can be served.
We therefore can write a request~$p$ as~$([o_p,d_p],[e_p,l_p])$.
Picking up or dropping off a passenger requires a \emph{service time} of~$t_{\textrm{s}}\in\N_0$ per passenger.
Further, a passenger may not leave the vehicle until arriving at their destination. 
We make the \emph{service promise} that the ride time of a passenger~$p$ may not exceed~$\alpha\geq 1$ times the \emph{direct time distance}~$t_{o_p,d_p}$, that is, the \emph{maximum ride time} of a passenger is~$\alpha\cdot t_{o_p,d_p}$.
The \emph{ride time} is measured from the end of pick-up to the beginning of drop-off.

In the \textsc{liDARP}, we further guarantee that if we pick up a passenger, the passenger is at all times transported towards their destination (regarding the sequence of stops). Consequently, a turn is only allowed for vehicles without passengers on board.
This so-called \emph{directionality} property~\cite{reiter_line-based_2024} constitutes the main difference between the \textsc{liDARP} and the `regular' \textsc{DARP}.

A \emph{tour} $r$ consists of a sequence of \emph{timestamped waypoints}, each waypoint being a pick-up/drop-off of a request with the timestamp corresponding to the start of the pick-up/drop-off.
In order for a tour to be \emph{feasible}, the timestamps need to adhere to the time constraints imposed by the time distances, as well as the service and turn times, the service promises, and the time windows which delimit the start of pick-ups and drop-offs.
Furthermore, at no point in time may there be more than~$c$ passengers in each vehicle and each request may only be served once. 
If we remove the timestamps from the waypoints, we obtain the \emph{route} underlying a tour. 
A route is \emph{feasible} if it can be complemented to a feasible tour by adding timestamps.

A given tour~$r$ can be decomposed into segments, called \emph{subtours}, with the vehicle turning precisely at the end of each subtours. 
Note that each subtour is, on its own, a tour, thus inheriting the feasibility definition.
We analogously define \emph{subroutes} for routes.
Similar to requests, we say that a subtour/subroute is~\emph{ascending} if the vehicle drives along the sequence of stops, and \emph{descending} if it drives the sequence in reverse. Thus, due to the directionality property, ascending requests must be served in ascending subtours and descending requests in descending subtours. We denote by~$\abs{r}$ the number of turns of a tour~$r$. Note that this corresponds to the number of subtours of~$r$ as a vehicle always turns at the end of a subtour.

A collection of tours~$R$ is \emph{feasible}, if each tour~$r\in R$ is feasible and each request is served at most once. 
We analogously define \emph{feasible collections of routes}.

A \emph{solution} to the \textsc{liDARP} is a feasible collection of up to~$k$ tours. The \textsc{liDARP} (as we consider it here) consists of determining a solution that maximizes the number of served requests.

Given a \textsc{liDARP} instance, the \textsc{MinTurn} problem determines the minimum over the 
largest number of turns a vehicle has to take in an optimal \textsc{liDARP} solution for this instance. 
\textsc{MinTurn} thus determines $\optturns:=\min_{R\in {\mathcal{R}^*}}\max_{r\in {R}}\abs{r}$ where~$\mathcal{R}^*$ is the set of optimal solutions for the given \textsc{liDARP} instance.

\paragraph{Conventions}
We assume that time starts at~$0$ and is integer.
We consider all cases as special cases of the \textsc{liDARP}: to omit the service promise, we set $\alpha=\infty$, the service and turn times can be omitted by setting~$t_{\textrm{s}}=0$ and $t_\textrm{turn}=0$, and the time windows can be disabled by setting~$e_p=0$ and~$l_p=\infty$ for all~$p\in P$.
In the case without time windows, a request~$p\in P$ is thus specified only by its origin and destination, i.e.,~$p = ([o_p, d_p])$.

We say that two requests \emph{overlap} if they are in the same direction and the intervals between their respective origin and destination are not interior disjoint.

We assume~$k$ and~$c$ to be bounded by the number of requests~$n$.
Given a positive integer~$n$, we use~$[n]$ as shorthand for~$\set{1, 2,\dots,n}$.

\section{Related Work} \label{sec:literature}
The Dial-a-Ride Problem has been extensively studied in the literature, with a focus on modelling approaches using mixed-integer linear programs and solution strategies including both exact strategies and heuristics. The surveys by Ho et~al.~\cite{HO2018395} and Vansteenwegen et~al.~\cite{vansteenwegen_survey_2022} provide a comprehensive overview.

In \cite{reiter_line-based_2024}, Reiter et~al.\ introduce the \emph{Line-Based Dial-a-Ride Problem (liDARP)}, where they aim to find a solution which maximizes a weighted sum of transported requests and saved distance (i.e., the difference between the sum of direct distances of transported passengers and the total distance driven by vehicles). The version of the \textsc{liDARP} studied here is a special case of that problem, as we consider only one of the objectives.
Reiter et~al.~\cite{reiter_line-based_2024} propose and compare three different mixed-integer linear formulations, including the \emph{subline-based} formulation which explicitly models sequences of turns for each vehicle.

The complexity of DARP on a line with \emph{makespan objective}, minimizing the \emph{completion time} (the time to serve all requests), has been addressed by a number of publications in the literature.
We summarize their findings in \cref{tab:lit_review:overview}, where $o = d$ denotes the setting where all requests' origins are equal to their destinations (equivalent to the \textsc{Travelling Salesperson Problem}).
Furthermore, some publications~(\cite{garey_computers_1979,tsitsiklis_special_1992}) consider \emph{individual} service times $t_s$ per request.

All publications listed in \cref{tab:lit_review:overview} fix the vehicles' starting positions, consider a \emph{closed} setting, where the vehicles have to return to their starting position at the end of the day, and require all~$n$ requests to be served.
\begin{table}[htb!]
    \centering
    \caption{Overview of known results for DARP on a line with makespan objective.}
    \begin{tabular}{c@{\hskip 8pt}c@{\hskip 8pt}c@{\hskip 8pt}cc@{\hskip 4pt}c@{\hskip 4pt}c}
        \toprule
        \#Veh. & Cap. & $o = d$ & Time Windows & Complexity & Ref. & Comment \\
        \midrule
        1 & 1 & $-$ & $-$ & polynomial & \cite{DBLP:journals/informs/PaepeLSSS04} & \\
        1 & 1 & $-$ & $e_p$ & \NP-complete & \cite{DBLP:journals/talg/BjeldeHDHLMSSS21} & Thm 7.6 \\
        1 & 2 & $-$ & $-$ & \NP-complete & \cite{guan_routing_1998} & \\
        1 & $\geq 2$ & $-$ & $-$ & \NP-complete & \cite{DBLP:journals/talg/BjeldeHDHLMSSS21} & Thm 7.8 \\
        1 & $\infty$ & $-$ & $-$ & polynomial & \cite{DBLP:journals/informs/PaepeLSSS04} & \\
        1 & $\geq 1$ & $\checkmark$ & $e_p$ & \NP-complete & \cite{tsitsiklis_special_1992} & individual $t_{\textrm{s}}$ \\
        1 & $\geq 1$ & $\checkmark$ & $[e_p, l_p]$ & \NP-complete & \cite{tsitsiklis_special_1992} &  \\
        2 & 1 & $\checkmark$ & $[e_p, l_p]$ & \NP-complete & \cite{garey_computers_1979} & individual $t_{\textrm{s}}$ \\
        $\geq 1$ & $c$ & $\checkmark$ & $l_p$ & polynomial & \cite{DBLP:journals/informs/PaepeLSSS04} & \\
        \bottomrule
    \end{tabular}
    \label{tab:lit_review:overview}
\end{table}

We note that de~Paepe \cite{DBLP:journals/informs/PaepeLSSS04} further showed that the setting with an arbitrary number of vehicles of fixed capacity~$c$, without time windows and where ${o = d}$, is also polynomially solvable under the objective of minimizing the sum of (weighted) completion times.

Further research has been conducted into examining the complexity of the Dial-a-Ride problem with individual loads per requests, minimizing the sum of driven distances on the half-line and line, as well as on the star, tree, circle, and $\R^d$ with the Euclidean metric \cite{archetti_complexity_2011,DBLP:journals/informs/PaepeLSSS04}.

Lastly, approximation algorithms for the related Vehicle Scheduling Problem on a line (L-VSP) have also been proposed. Karuno et~al. \cite{karuno_better_2002} developed a $\frac{3}{2}$-approximation algorithm for the closed L-VSP with a single vehicle and time windows, under minimizing completion times, where the starting vertex is fixed. Allowing for arbitrary starting vertices, Gaur et~al.~\cite{GAUR200387} develop a $\frac{5}{3}$-approximation for the L-VSP. Under the objective of minimizing makespan, with an arbitrary starting position, Yu and Liu \cite{yu_liu_vsp} develop a $\frac{3}{2}$-approximation for the closed, and a $\frac{5}{3}$-approximation algorithm for the open variants.

Considering the Multi-Vehicle Scheduling Problem (MVSP) on a line, Karuno and Nagamochi \cite{karuno_nagamochi_MVSP} present a $2$-approximation algorithm to minimize the total completion time for a fixed number of vehicles.

\section{Our Contribution} \label{sec:contribution}
In this paper, we study the complexity of the \textsc{liDARP} and the novel \textsc{MinTurn} problem. 
Unlike the \textsc{DARP} on a line studied in the literature, we consider an \emph{open} setting, where the vehicles do not have to return to their (arbitrary) starting positions, and maximize the number of served passengers as an objective. 
Note that most of our results can be transferred to the closed setting.

We consider different instance parameters and characteristics: the number of vehicles~$k$ and their capacity~$c$, as well as the presence of time windows, shortcuts, turn times, the service promise, and the service time.
Our complexity results for the \textsc{MinTurn} problem are summarized in~\cref{tab:contribution:min_turn} and novel results for the \textsc{liDARP} are given in~\cref{tab:contribution:lidarp}.
Note that whether there is a turn time or not appears to be irrelevant for the complexity of the problems: we have no turn time in all hardness results, while all algorithmic results hold for arbitrary turn times. 
We further show that it is strongly \NP-hard to approximate \textsc{MinTurn} with a factor better than~$3$ (\cref{cor:mt-approx-nphard}).

\begin{table}[htb!]
    \centering
    \begin{threeparttable}
    \caption{Overview of novel results for the \textsc{MinTurn} problem presented here.}
    \begin{tabular}{c@{\hskip 5pt}c@{\hskip 5pt}c@{\hskip 5pt}c@{\hskip 5pt}c@{\hskip 5pt}c@{\hskip 5pt}c@{\hskip 5pt}c}
        \toprule
        \multirow{2}*{\#Veh.}  & \multirow{2}*{Cap.} & Time & \multirow{2}*{Shortcuts} & Service & Service & \multirow{2}*{Complexity} & \multirow{2}*{Ref.}  \\
         & & Windows & & Promise & Time & & \\
        \midrule
        $\geq 1$ & $\geq 1$  & $-$ & $-$ & \checkmark & $-$ & polynomial & \cshref{thm:mt-poly-variants} \\
        $\geq 1$ & $\geq 1$  & $-$ & \checkmark & $-$ & \checkmark & polynomial & \cshref{thm:mt-poly-variants} \\
        $\geq 1$ & $1$  & $-$ & \checkmark & \checkmark & \checkmark & polynomial & \cshref{thm:mt-poly-variants} \\
        $\geq 1$ & $\geq 2$ &  $-$ & $-$ & \checkmark & \checkmark & strongly \NP-hard & \cshref{thm:sp-st-nphard} \\
        $\geq 1$ & $\geq 2$ &  $-$ & \checkmark & \checkmark & $-$ & strongly \NP-hard & \cshref{thm:sp-sc-nphard} \\
        $\geq 1$ & $\geq 1$ &  \checkmark & $-$ & $-$ & $-$ & strongly \NP-hard & \cshref{thm:mt-tw-nphard} \\
        \bottomrule
    \end{tabular}
    \label{tab:contribution:min_turn}
    \end{threeparttable}
\end{table}

\begin{table}[htb!]
    \centering
    \caption{Overview of novel results for the \textsc{liDARP} problem presented here.}
    \begin{tabular}{c@{\hskip 5pt}c@{\hskip 5pt}c@{\hskip 5pt}c@{\hskip 5pt}c@{\hskip 5pt}c@{\hskip 5pt}c@{\hskip 5pt}c}
        \toprule
        \multirow{2}*{\#Veh.}  & \multirow{2}*{Cap.} & Time & \multirow{2}*{Shortcuts} & Service & Service & \multirow{2}*{Complexity} & \multirow{2}*{Ref.}  \\
         & & Windows & & Promise & Time & & \\
        \midrule
        $\geq 1$ & $\geq 1$  & $-$ & \checkmark & \checkmark & \checkmark & polynomial & \cshref{thm:poly-lidarp} \\
        $\geq 1$ & $\geq 1$ & \checkmark & $-$  & $-$ & $-$ & strongly \NP-hard & \cshref{thm:mt-tw-nphard}
        \\
        \bottomrule
    \end{tabular}
    \label{tab:contribution:lidarp}
\end{table}

\newpage
\section{Polynomially Solvable Cases}\label{sec:pol}
In this section, we characterize the cases in which the \textsc{liDARP} and \textsc{MinTurn} problem are polynomially solvable.
We begin by showing that it suffices to consider feasible routes, as we can efficiently transform them into feasible tours.

\begin{restatable}[\restateref{lem:route-to-tour}]{lemma}{RouteToTour}
    \label{lem:route-to-tour}
    Given a route, we can check in polynomial time whether it is feasible and, if so, complement it to a feasible tour.
    If there are no time windows, this can even be done in linear time. 
    If additionally, there is no service promise, the route is feasible as long as it respects capacities.
\end{restatable}

Without time windows it even suffices to find feasible subroutes, as any route obtained by \emph{joining} feasible routes, i.e., concatenating the sequences of waypoints of the routes, is feasible. 

\begin{restatable}[\restateref{lem:route-joining}]{lemma}{RouteJoining}
    \label{lem:route-joining}
    Consider a \textsc{liDARP} instance without time windows and a feasible collection~$R$ of routes. Joining all routes in~$R$ (in arbitrary order) yields a feasible route.
\end{restatable}

We now use these lemmas, to show that the \textsc{liDARP} is 
polynomially solvable in the absence of time windows, by constructing a feasible tour that serves all requests consecutively.

\begin{restatable}[\restateref{thm:poly-lidarp}]{theorem}{PolyLiDARP}
    \label{thm:poly-lidarp}
    If there are no time windows, all requests can be served. A solution for the \textsc{liDARP} serving all requests can be computed in linear time.
\end{restatable}

As we see later in~\cref{thm:mt-tw-nphard}, the \textsc{liDARP} is \NP-hard as soon as we have time windows. 
We therefore now focus on the \textsc{MinTurn} problem. 
We begin by showing that if we have already determined the (number of) subroutes needed to serve all requests, we can efficiently compute~$\optturns$ for the \textsc{MinTurn} problem.

\begin{lemma}
    \label{lem:mt-subroute-combination}
    Consider an instance of the \textsc{MinTurn} problem without time windows. Let~$a$ ($b$) be the smallest number of feasible ascending (descending) subroutes needed to serve all ascending (descending) requests. Assume w.l.o.g. that $a\geq b$. 
    Then, $\optturns = \max\set*{\ceils*{\frac{a+b}{k}},2\cdot\ceils*{\frac{a}{k}}-1}$.
\end{lemma}
\begin{proof}
    We observe that we can create a feasible collection of~$k$ routes serving all requests by alternatingly joining ascending and descending subroutes into routes of length~$2\ceils*{\frac{a}{k}}$, using each subroute once and adding \emph{artificial subroutes}, that do not serve requests, in case there not enough subroutes. The resulting routes are feasible according to~\cref{lem:route-joining}, as all subroutes, including the artificial subroutes, are feasible.

    Furthermore, we can prove two lower bounds on the number of turns per vehicle: 
    First, by the pigeonhole principle, there has to be a route consisting of at least~$\ceils*{\frac{a+b}{k}}$ subroutes.
    Second, as there must be a route containing at least~$\ceils{\frac{a}{k}}$ ascending subroutes, and ascending and descending subroutes alternate, this route contains at least~$2\ceils*{\frac{a}{k}}-1$ subroutes.
    The lower bound is thus~$\max\set*{\ceils*{\frac{a+b}{k}},2\ceils*{\frac{a}{k}}-1}$.

    If~$\ceils{\frac{a+b}{k}}=2\ceils{\frac{a}{k}}$, the upper and lower bound coincide. Otherwise, it holds that~$\ceils{\frac{a+b}{k}}\leq 2\cdot\ceils{\frac{a}{k}}-1$. 
    In this case, we add~$d$ artificial descending subroutes such that~$\frac{a+b+d}{k}=2\ceils*{\frac{a}{k}}-1$. 
    We then combine the subroutes into a feasible collection of~$k$ routes serving all requests, each with~$2\ceils*{\frac{a}{k}}-1$ turns.
    For this, we first add~$\ceils{\frac{a}{k}}-1$ ascending subroutes as well as~$\ceils{\frac{a}{k}}-1$ descending subroutes to each route. This leaves exactly~$k$ subroutes unassigned, as~$a+b+d=\parens*{2\ceils{\frac{a}{k}}-1}\cdot k=2k\parens{\ceils*{\frac{a}{k}}-1}+k$. We can thus assign each of these~$k$ subroutes to a separate route. Thus, each route is assigned~$2\ceils*{\frac{a}{k}}-1$ subroutes, with a route either containing~$\ceils*{\frac{a}{k}}$ ascending or~$\ceils*{\frac{a}{k}}$ descending subroutes. By alternating the ascending and descending subroutes in a route, we obtain routes with~$2\ceils*{\frac{a}{k}}-1$ turns. According to~\cref{lem:route-joining}, these routes form a feasible collection. 
    \qed
\end{proof}

To solve the \textsc{MinTurn} problem in the absence of time windows, we thus need to determine the minimum number of feasible ascending and descending subroutes needed to serve all requests. We now show that, if the feasibility of subroutes is determined by the capacity, this can be done in polynomial time.

\begin{lemma}\label{lem:coloring}
    Consider a \textsc{MinTurn} instance where subroutes are already feasible if they respect the capacity. Let~$\chi$ be the maximum number of pairwise overlapping ascending (descending) requests. The minimum number of feasible ascending (descending) subroutes needed to serve all ascending (descending) requests is~$\ceils{\chi/c}$. Determining~$\chi$ is possible in polynomial time.
\end{lemma} 
\begin{proof}
    The assignment of ascending (descending) requests to seats in subroutes corresponds to coloring the conflict graph of the requests, as no two overlapping requests may occupy the same seat. As the conflict graph is an interval graph, the chromatic number~$\chi$ corresponds to the maximum number of pairwise overlapping requests~\cite{HOUGARDY20062529} and can be determined in polynomial time.
    Thus, if~$\chi$ seats are needed,~$\ceils*{\chi/c}$ subroutes are necessary to serve all requests.
    \qed
\end{proof}

These lemmas imply that a \textsc{MinTurn} instance without time windows is solvable in polynomial time if the feasibility of subroutes is determined solely by the capacity constraints. We use this insight to characterize the cases in which the \textsc{MinTurn} problem is polynomially solvable.

\begin{restatable}[\restateref{thm:mt-poly-variants}]{theorem}{MTPolyVariants}
    \label{thm:mt-poly-variants}
    Consider an instance of \textsc{MinTurn}. Let~$a$ ($b$) be the maximum number of pairwise overlapping ascending (descending) requests and assume w.l.o.g. $a\geq b$. In the following cases of the \textsc{MinTurn} problem, we have $\optturns = \max\set{\ceils*{\frac{a+b}{k}},2\cdot\ceils{\frac{a}{k}}-1}$, which can be determined in polynomial time:
    \begin{enumerate}
     \item without time windows and  without service promise   
     \item without time windows, without shortcuts, and without service times
    \item without time windows and with a capacity of~$1$
    \end{enumerate}
\end{restatable}

\newpage
\section{Hardness Results}\label{sec:hard}
In this section, we show that all remaining \textsc{liDARP} and \textsc{MinTurn} cases are strongly \NP-hard, using reductions from \textsc{3-Partition}, which is well known to be strongly \NP-hard, see~\cite{garey_computers_1979}.

\begin{definition}[\textsc{3-Partition}~\cite{garey_computers_1979}]
    \label{def:3-partition}
    Given a finite set~$S$ of~$n=3m$ positive integers as well as a bound~$T\in\N$ such that~$\sum_{s\in S}s=mT$ and~$T/4<s<T/2$ for all~${s\in S}$, is there a partition of~$S$ into~$m$ disjoint sets~$S_1,\dots,S_m$ such that~$\sum_{s\in S_j}s=T$ for all~$j\in [m]$?
\end{definition}

We begin by showing hardness of \textsc{MinTurn} even in the absence of time windows and shortcuts.

\begin{restatable}[\restateref{thm:sp-st-nphard}]{theorem}{SPSTHardness}
    \label{thm:sp-st-nphard}
    The problem \textsc{MinTurn} is strongly \NP-hard for all~$k\geq 1$ and $c\geq 2$ even without time windows, shortcuts, and turn time. 
\end{restatable}
\begin{proof}
    We begin by showing the reduction from \textsc{3-Partition} for~$k=1$ and $c=2$. Let~$(S,m,T)$ be an instance of \textsc{3-Partition} and~$S=\set{s_1,\dots,s_n}$. 
    
    We create an instance of \textsc{MinTurn} such that the ascending subroutes in an optimal solution of the \textsc{liDARP} with minimum turns per route
    correspond to a $3$-partition of~$S$. An example of such an instance can be seen in~\cref{fig:sp-st-hardness}. 
    
    We begin by having stops~$H=\angles{1,\dots,4+4Tn}$ with unit distance between neighboring stops. 
    We create four types of requests: for each~$i\in[n]$ we create~$s_i$ \emph{value requests}~$P_\textrm{V}^i=\set{([4+T(i-1)+(j-1),4+T(i-1)+j])\mid j\in[s_i]}$ and~$m$ \emph{plug requests}~$P_\textrm{P}^i=P_\textrm{LP}^i\cup\set{p_\textrm{P}^i}$. The plug requests consist of~$m-1$ \emph{long plug requests}~$P_\textrm{LP}^i$, each being~$([4+T(i-1),4+Ti])$, and one \emph{short plug request} $p_\textrm{P}^i=([4+T(i-1)+s_i,4+Ti])$. 
    We also create $m$ \emph{promise requests} $P_\textrm{SP}$, each being~$([1,4+4Tn])$, which are used in combination with the service promise to ensure that the number of value requests in a subroute does not exceed~$T$. Lastly, we create~$m$ \emph{filter requests}~$P_\textrm{F}$, each being~$([2,3])$, which are used to ensure that each subroute contains exactly one promise request. We set the service time~$t_{\textrm{s}}=1$ and the service promise~$\alpha=1+b/a$ with~$a=3+4Tn$ and~$b=2(1+T+n)$. Note that~$b<a$ and thus~$\alpha<2$.
    
    We now show that~$S$ has a $3$-partition if and only if~$\optturns=2m-1$ for the constructed \textsc{MinTurn} instance.
        
    \begin{figure}
        \centering
        \includegraphics{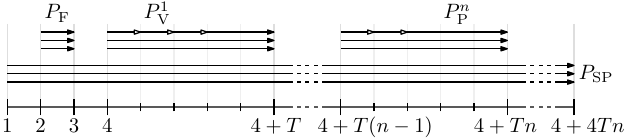}
        \caption{
        A \textsc{MinTurn} instance constructed from a \textsc{3-Partition} instance with~$m=3$ and~$T=5$, as well as~$s_1=3$ and~$s_n=2$. The arrows represent the requests, with the white tipped arrows being value requests.
        \label{fig:sp-st-hardness}}
    \end{figure}   
    
    From~\cref{thm:poly-lidarp}, we know that in an optimal \textsc{liDARP} solution to the constructed instance all requests are served. As we see in Observation~$1$, each filter request has to be served by a different subroute, requiring at least~$m$ ascending subroutes to serve all requests. We thus need at least~$2m-1$ turns to serve all requests.
    Assume that we have an optimal solution to the \textsc{liDARP} that uses~$2m-1$ turns. Its route must thus contain~$m$ ascending subroutes.
    
    \emph{Observation 1:} Each ascending subroute serves exactly one filter and one promise request. Indeed, as the service promise~$\alpha$ is less than~$2$ and the service time~$t_{\textrm{s}}$ is~$1$, the maximum ride time of a filter request is less than~$2$, due to its direct time distance being~$1$, and would thus be exceeded if another filter request is served by the same subroute. Since the~$m$ ascending subroutes serve all requests, each ascending subroute must thus contain exactly one filter request. It follows that each ascending subroute must also contain exactly one promise request, since promise and filter request overlap.
    
    \emph{Observation 2:} Due to the service promise, the maximum ride time of a promise request is at most~$b$ more than the direct time distance. As all other requests lie between the origin and destination of a promise request and the service time is~$1$, at most~$1+T+n$ requests can be transported in a subroute besides a promise request. According to Observation~$1$, one of these requests must be a filter request.
    
    \emph{Observation 3:} Each subroute also has to serve exactly one plug request for each~$i\in[n]$, as there are~$m$ such requests and they overlap each other as well as the promise requests.

    \emph{Observation 4:} Combining Observations~$2$ and $3$, we conclude that each subroute may serve up to~$T$ value requests. As the total number of value requests is~$mT$, this means that each subroute transports exactly~$T$ such requests.

    \emph{Observation 5}: All requests in~$P_\textrm{V}^i$ must be served by the same subroute, the one that serves the short plug request~$p_\textrm{P}^i$, as all the other subroutes contain long plug requests that overlap all requests in~$P_\textrm{V}^i$. That is, for each~$i\in[n]$ we have exactly one subroute that serves the~$s_i$ value requests~$P_\textrm{V}^i$. In combination with Observation 4 we conclude that, for each subroute~$r_j$, we have an index set~$I_j$, such that all value requests in~$\bigcup_{i \in I_j} P_\textrm{V}^i$ are served by the subroute~$r_j$ and $\sum_{i \in I_j} s_i=T$. 
    Setting~$S_j:=\{s_i\mid i\in I_j\}$, we thus obtain a valid $3$-partition of~$S$.

    Conversely, if there exists a $3$-partition~$S_1,\dots,S_m$ of~$S$, we create~$m$ ascending subroutes~$r_1,\dots,r_m$ and assign for each~$s_i\in S_j$ for~$j\in[m]$ the value requests in~$P_\textrm{V}^i$ as well as the short plug request in~$p_\textrm{P}^i$ to the~subroute~$r_j$. To each of the remaining ascending subroutes, we assign one of the long plug requests from $P_\textrm{LP}^i$. Furthermore, we assign one filter and one promise request to each ascending subroute. In this way, all requests are assigned to a subroute. Additionally, the capacities are respected, as no more than~$2$ requests pairwise overlap in a subroute. 
    By adding~$m-1$ artificial descending subroutes, we connect the ascending subroutes and obtain a route with~$2m-1$ turns. For this route to be feasible it remains to show, according to~\cref{lem:route-joining}, that the subroutes are feasible.
    The only requests which are not transported directly are the promise requests. By construction, each subroute serves, besides the promise request, one filter,~$n$ plug, and~$T$ value requests. Therefore, 
    the delay in the ride time of a promise request is~$2(1+T+n)$, which is precisely the allowed delay by the service promise, and the subroutes are thus feasible.
    
    For higher capacities and more vehicles we duplicate the promise requests, such that they fill the added seats. Apart from a small adjustment of the service promise and subsequently the direct time distance of promise requests, the construction and correctness are analogous.

    Constructing these instances takes pseudo-polynomial time. As \textsc{3-Partition} is strongly \NP-hard, it follows that the \textsc{MinTurn} problem is strongly \NP-hard.
    \qed
\end{proof}

The presented reduction can be adapted to show strong \NP-hardness for the case without service times but instead with shortcuts, by encoding the values~${s\in S}$ into detours of the line that can be shortcut by a subroute if it does not serve the corresponding requests.

\begin{restatable}[\restateref{thm:sp-sc-nphard}]{theorem}{SPSCHardness}
    \label{thm:sp-sc-nphard}
    The problem \textsc{MinTurn} is strongly \NP-hard for all $k\geq 1$ and $c\geq 2$ even without time windows, service times and turn times.
\end{restatable}
\begin{proof}
    We begin by proving hardness for~$k=1$ and~$c=2$ before extending it to higher values. 
    Let~$(S,m,T)$ be an instance of \textsc{3-Partition} and~$S=\set{s_1,\dots,s_n}$. 
    
    As we use shortcuts, we start by describing the layout of the line, which can be seen schematically in~\cref{fig:sp-sc-hardness}.

    \begin{figure}
        \centering
        \includegraphics{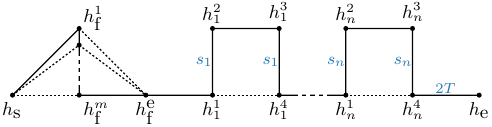}
        \caption{
            The layout of the stops constructed for a \textsc{3-Partition} instance. The line is the continuous path, while shortcuts are represented by (dotted) lines. The distances between neighboring stops are given by (blue) labels, with the exception of distances of~$1$, which are omitted.
        }
        \label{fig:sp-sc-hardness}
    \end{figure}
    
    The line starts at stop~$h_\textrm{s}$ then contains (in the order in which they are listed here) a sequence of stops~$h_\textrm{f}^j$ for~$j\in [m]$, a stop~$h_\textrm{f}^\textrm{e}$, for each~$s_i\in S$ a sequence of stops~$h_i^1,h_i^2,h_i^3,h_i^4$, and the final stop~$h_\textrm{e}$.
    For each~$i\in [n]$, the distance between~$h_i^1$ and~$h_i^2$ as well as between~$h_i^3$ and~$h_i^4$ is~$s_i$. The distance between~$h_n^4$ and~$h_\textrm{e}$ is~$2T$.
    For all other pairs of subsequent stops, the distances are~$1$. 
    We have a number of shortcuts: the direct distance from $h_\textrm{s}$ to $h_\textrm{f}^j$ and from $h_\textrm{f}^j$ to $h_\textrm{f}^\textrm{e}$ is $1$ for all $j\in[m]$. Furthermore, the direct distance between~$h_i^1$ and~$h_i^4$ is~$1$ for all~$i\in [n]$.
    The~$m$ \emph{promise requests}~$P_\textrm{SP}$ originate at~$h_\textrm{s}$ and end at~$h_\textrm{e}$.
    Each of the~$m$ \emph{filter requests}~$p_\textrm{F}^j\in P_\textrm{F}$ for $j \in [m]$ starts at stop~$h_\textrm{f}^j$ and goes to~$h_\textrm{f}^\textrm{e}$. 
    For each~$i\in [n]$, a \emph{value request}~$p_i$ is added that originates at~$h_i^2$ and goes to~$h_i^3$.
    The service promise is set to~$1+b/a<2$ with~$a=2+2n+2T$ and~$b=2T$.
    
    It then holds that~$S$ has a $3$-partition if and only if
    $\optturns=2m-1$ for the constructed \textsc{MinTurn} instance.
    The proof is similar to the proof of~\cref{thm:sp-st-nphard}.

    Differing from the reduction in~\cref{thm:sp-st-nphard}, we can construct this instance in polynomial time, as our number of stops does not depend on~$T$. It follows that the \textsc{MinTurn} problem is strongly \NP-hard.
    \qed
\end{proof}

By adapting the reductions to use~$m$ vehicles instead of one, we can show that the \textsc{MinTurn} problem is strongly \NP-hard to approximate with a factor better than~$3$. This leads to the following corollary.

\begin{restatable}[\restateref{cor:mt-approx-nphard}]{corollary}{InApproximable}
    \label{cor:mt-approx-nphard}
    The problem \textsc{MinTurn} is strongly \NP-hard to approximate with a factor better than~$3$ for all $c\geq 2$ in the following cases:
    \begin{enumerate}
        \item without time windows, shortcuts, and turn times, and 
        \item without time windows, service times, and turn times.
    \end{enumerate}
\end{restatable}

We now show that as soon as we consider time windows, both the \textsc{liDARP} and \textsc{MinTurn} problem become \NP-hard for arbitrary values of~$k$ and~$c$.

\begin{restatable}[\restateref{thm:mt-tw-nphard}]{theorem}{MTTWHardness}
    \label{thm:mt-tw-nphard}
    Both \textsc{liDARP} and \textsc{MinTurn} are strongly \NP-hard for all $k\geq 1$ and~$c\geq 1$ even without shortcuts, service promise, service times, and turn times.
\end{restatable}
\begin{proof}[sketch]
    Again, we use \textsc{3-Partition} for the reduction. The main idea of the proof for~$k=1$ and~$c=1$ is to translate the values~$s\in S$ into value requests that need time~$2s$ to be served. We then use \emph{separator} requests to create~$m$ time intervals of length~$2T$ during which the value requests must be served if all requests are served. Thus, the assignment of value requests to time intervals corresponds to a~$3$-partition of~$S$.
    \qed
\end{proof}

\section{Parameterized Algorithms}
\label{sec:parameterized-algos}
Seeing as the \textsc{liDARP} and \textsc{MinTurn} problem are \NP-hard, we now provide parameterized algorithms for both. Recall that an algorithm is \emph{fixed-parameter tractable} (\FPT) w.r.t.\ a parameter~$k$ if its runtime is~$f(k)\cdot n^{O(1)}$ for some computable function~$f$, where $n$ is the size of the input. 
An algorithm is \emph{slice-wise polynomial} (\XP) w.r.t.\ a parameter~$k$ if its runtime is in~$O(n^{f(k)})$ for some computable function~$f$. 
When analyzing the runtime, we use the~$O^*$ notation, which suppresses polynomial factors in the input size, i.e., a function~$g(n)$ is in~$O^*(f(n))$ if there is a polynomial~$p$ such that~$g(n)\in O(f(n)\cdot p(n))$.

As we have shown the \textsc{liDARP} and \textsc{MinTurn} problem to be \NP-hard for constant~$k$ and~$c$, we have to consider more parameters to obtain parameterized algorithms. We therefore use the number of stops~$h$ as well as the maximum time~$t:=\max_{p\in P}l_p+1$.

\begin{restatable}[\restateref{thm:minturn-fpt}]{theorem}{MinTurnFPT}
    \label{thm:minturn-fpt}
    There exists an \FPT-algorithm for \textsc{MinTurn} as well as \textsc{liDARP} parameterized by~$k$,~$c$,~$h$ and~$t$, with a runtime in $O^*((h^2\cdot t^3\cdot c\cdot k)^{2\cdot t\cdot c\cdot k})$.
\end{restatable}
\begin{proof}[sketch]
    To prove this result, we enumerate all routes that could be part of a \textsc{liDARP} solution. As we can bound the number of requests that a single vehicle can serve in time $t$ by $t\cdot c$, we can enumerate all feasible routes in time $O^*(n^{2\cdot t\cdot c})$, using the \emph{event-based graph}~\cite{DBLP:journals/eor/GaulKS22,reiter_line-based_2024}. 
    For all collections of up to~$k$ routes, we find the one which serves the most requests and minimizes the maximum turns per route.
    This results in a runtime of~$O^*(n^{2\cdot t\cdot c\cdot k})$.

    To obtain an \FPT-runtime, we slightly modify the algorithm such that it reduces inputs to a predetermined size, using the observation that there are at most~$h^2\cdot t^2$ distinct requests which can be served at most~$t\cdot c\cdot k$ times each.
    \qed
\end{proof}

Note that the \FPT-algorithm can be adapted easily to work for several other objectives and restrictions, such as maximizing the weighted sum of saved distance and transported requests, as used in~\cite{reiter_line-based_2024}, or minimizing the makespan while serving all requests, as studied in most complexity papers on \textsc{DARP} on a line, see \cref{tab:lit_review:overview}.

However, the case without time windows poses some difficulties. In this paper we treat this as a special case of the general case by setting~$l_p=\infty$, which implies~$t=\infty$. 
Determining a better bound for~$l_p$ is of no use, as it would depend on~$n$.
Thus, we now devise an \XP-algorithm for the case without time windows whose running time is parameterized by~$c$ and~$h$.

\begin{restatable}[{\restateref{thm:minturn-xp-no-tw}}]{theorem}{MinTurnXPNoTW}
   \label{thm:minturn-xp-no-tw}
   There is an \XP-algorithm for the \textsc{MinTurn} problem without time windows, parameterized by~$c$ and~$h$, with runtime~$O^*(n^{h^2}\cdot h^{4\cdot c\cdot h})$.
\end{restatable}
\begin{proof}[sketch]
    We interpret finding the minimum number of feasible subroutes needed to serve all ascending (descending) requests as a \textsc{Multiset Multicover} problem and apply an algorithm proposed by Hua~et~al.~\cite{DBLP:journals/tcs/HuaWYL10} to solve it. Applying~\cref{lem:mt-subroute-combination}, we obtain the value of $\optturns$ for \textsc{MinTurn} in time~$O^*(n^{h^2}\cdot h^{4\cdot c\cdot h})$.
    \qed
\end{proof}

\section{Conclusion}
\label{sec:conclusion}
We introduce the \textsc{MinTurn} problem and characterize the boundary between polynomial solvability and \NP-hardness for the \textsc{MinTurn} and \textsc{liDARP} problem according to instance specifics, including time windows, shortcuts, service promise, and service times. We also show that the \textsc{MinTurn} problem is strongly \NP-hard to approximate with factor better than~$3$.
We then provide an \FPT-algorithm for the \textsc{MinTurn} and \textsc{liDARP} problem, parameterized by~$k$,~$c$,~$h$ and~$t$, which also works for other objectives and restrictions. 
Finally, we present an \XP-algorithm for the \textsc{MinTurn} problem without time windows parameterized by~$c$ and~$h$.

An interesting topic for further research would be to analyze the \textsc{liDARP} and \textsc{MinTurn} problem for other objectives and restrictions, such as minimizing the sum of turns (over all vehicles), or allowing an unlimited number of vehicles. It also remains open, whether the parameterized algorithms can be improved, such as if there is an \FPT-algorithm parameterized by~$c$ and~$h$.
For practical applications, such as the subline-based formulation~\cite{reiter_line-based_2024}, it may be interesting to develop heuristics and approximation algorithms for the \textsc{MinTurn} problem.

\newpage
%
%
%
\bibliographystyle{splncs04}
\bibliography{turns-lidarp}

\appendix

\section*{Appendix}

\subsection*{Completion of Proofs for the \emph{Polynomially Solvable Cases}}

\RouteToTour*
\label{lem:route-to-tour*}
\begin{proof}
    Testing if a route can be complemented to a feasible tour entails checking that no request is picked-up twice, the capacity is not violated, and that the time windows and service promise can be kept while respecting the time constraints imposed by the time distances, as well as the service and turn times. As shown by~\cite{haugland_feasibility_2010}, this can be done in polynomial time.

    We will now show that if there are no time windows, the check can be done in linear time. 
    Let~$r=\angles{w_1,\dots,w_m}$ be the given route consisting of~$m$ (untimed) waypoints. 
    For a timed waypoint~$w$, let~$t(w)$ be the timestamp and~$h(w)$ the stop.

    We now construct a tour~$r^\mathrm{T}=\angles{w_1^\mathrm{T},\dots,w_m^\mathrm{T}}$ corresponding to the route~$r$. We set the timestamps of the waypoints as follows: for the first waypoint~$w_1^\mathrm{T}$ we set~$t(w_1^\mathrm{T})=0$. For each subsequent waypoint~$w_j^\mathrm{T}$, we set~$t(w_j^\mathrm{T})=t(w_{j-1}^\mathrm{T})+\gamma$, where~$\gamma$ is the time constraint imposed by the time distance, as well as the service and turn time. 
    We have~$\gamma=t_{h(w_{j-1}),h(w_j)}+t_s+\delta\cdot t_\mathrm{turn}$, with~$\delta$ being the number of turns necessary between the two waypoints. If the direction of the vehicle is different at the two waypoints, one turn is needed. Otherwise, the vehicle has the same direction at both waypoints. Then, if the stop~$h(w_{j-1})$ comes after~$h(w_j)$ in the direction of the vehicle, two turns are needed, as the vehicle has to turn twice to drive back. Otherwise, no turn is needed.
    
    In the tour~$r^\mathrm{T}$, the drive time of each served request is minimized for the route, as the time between two waypoints is always set to the time constraint imposed by the time distance, as well as the service and turn time, which is consequently fulfilled by construction. Thus, if the tour violates the service promise, there is no feasible tour corresponding to the route. 
    Constructing the tour and checking whether the service promise is fulfilled can be done in linear time, each by a linear sweep along the route.    
    It remains to check the capacity constraints. We can accomplish this with another linear sweep along the route, maintaining the number of passengers currently transported.

    As we have seen, we can always construct a tour that respects the time constraints imposed by the time distances, service, and turn times. Thus, if there are no time windows and no service promise, the feasibility of a route depends only on the capacity constraints.
    \qed
\end{proof}

\RouteJoining*
\label{lem:route-joining*}
\begin{proof}
    Let~$r$ be the resulting route. Clearly,~$r$ serves each request at most once and respects the capacity constraint.
    It remains to show that~$r$ is feasible, i.e., there is a corresponding feasible tour. Let~$R^\mathrm{T}=\set{r_1^\mathrm{T},\dots,r_k^\mathrm{T}}$ be feasible tours corresponding to the routes~$R=\set{r_1,\dots,r_k}$. We now construct a tour~$r^\mathrm{T}$ that corresponds to the route~$r$.

    For a tour~$r^T_i=\angles{w_1^i,\dots,w_{m(i)}^i}$ and~$j\in[m(i)]$, let~$t(w_j^i)$ be the timestamp and~$h(w_j^i)$ the stop of the timed waypoint~$w_j^i$. 

    We construct~$r^\mathrm{T}$ as follows:  
    we first add the waypoints of~$r_1^\mathrm{T}$ to the tour. To differentiate, we write~$\bar{w}^i_j$ for a waypoint in~$r^\mathrm{T}$ corresponding to the waypoint~$w_j^i$. We set~$t(\bar{w}_1^1)=0$ and, for each subsequent waypoint~$\bar{w}_j^1$, we set~$t(\bar{w}_j^1)=t(\bar{w}_{j-1}^1)+t(w_j^1)-t(w_{j-1}^1)$, thus retaining the time difference to the preceding waypoint. For each subsequent tour~$r_i^\mathrm{T}$, we proceed analogously, retaining for each waypoint except the first the time difference to the preceding waypoint. For the first waypoint~$\bar{w}_1^i$, we set~$t(\bar{w}_1^i)=t(\bar{w}_{m(i-1)}^{i-1})+\gamma$, where~$\gamma$ is the time constraint imposed by the time distance, as well as the service and turn time. We have~$\gamma=t_{h(\bar{w}_{m(i-1)}^{i-1}),h(\bar{w}_1^i)}+t_s+\delta\cdot t_\mathrm{turn}$, with~$\delta$ being the number of turns necessary between the two waypoints. If the tours~$r_{i-1}$ and~$r_i$ have different directions one turn is needed. Otherwise, both tours have the same direction. If the stop~$h(\bar{w}_{m(i-1)}^{i-1})$ is after~$h(\bar{w}_1^i)$ in the direction of the tours, two turns are needed, as the vehicle has to drive back. Otherwise, no turn is needed.

    After applying this procedure to all tours, the time constraints~$\gamma$, imposed by the distances as well as the service and turn times, are fulfilled by construction. Further, the service promise is still fulfilled, as ride times did not change in comparison to the original tours, since we did not modify the time differences inside a tour. Thus the tour~$r^\mathrm{T}$, and therefore the joined route~$r$, is feasible.
    \qed
\end{proof}

\PolyLiDARP*
\label{thm:poly-lidarp*}
\begin{proof}
    Clearly, a route that serves a single request is feasible. Thus, serving each request in a different route is a feasible collection of routes. By~\cref{lem:route-joining}, joining these routes into a single route yields a feasible route that serves all requests. According to~\cref{lem:route-to-tour}, we can transform this route into an optimal tour and thus an optimal \textsc{liDARP} solution in linear time.
    \qed
\end{proof}

\MTPolyVariants*
\label{thm:mt-poly-variants*}
\begin{proof}
    In all cases, a subroute is feasible if and only if we do not violate the capacity:

    \begin{enumerate}
    \item As there are no time windows and no service promise, according to~\cref{lem:route-to-tour}, the feasibility of a subroute depends only on the capacity constraints.
    \item Since there are no service times as well as no shortcuts, picking-up and dropping-off a passenger does not cause a delay for the passengers already in the vehicle. Thus, each passenger always arrives after their direct time distance at their destination, regardless of other passengers transported by the vehicle. Thus, the service promise is always kept and this case is equivalent to there being no service promise.
    \item Since the capacity of the vehicles is~$1$, each passenger has to be transported individually, and is therefore driven directly to their destination. Thus, the service promise is always kept and this case is equivalent to there being no service promise.
    \end{enumerate}

    In order to obtain a minimum number of subroutes, we apply~\cref{lem:coloring} to obtain in polynomial time the minimum number~$a$ of feasible ascending and~$b$ of feasible descending subroutes needed to serve all requests. Using~\cref{lem:mt-subroute-combination}, we obtain that~$\optturns=\max\set{\ceils*{\frac{a+b}{k}},2\cdot\ceils{\frac{a}{k}}-1}$.
    \qed
\end{proof}

\newpage
\subsection*{Completion of Proofs for the \emph{Hardness Results}}

\SPSTHardness*
\label{thm:sp-st-nphard*}
\begin{proof}
    The reduction for~$k=1$ and~$c=2$ is in the main part of the paper. We now extend that reduction, first to higher capacities and then to higher numbers of vehicles.

    For higher capacities, we add more promise requests, such that there are in total~$(c-1)\cdot m$. We further generalize the destination of the promise requests to~$4+4Tn+(c-2)$. To compensate for the delay experienced by the promise requests due to having to serve other promise requests in the same subroute, we generalize the service promise to~$\alpha=1+b/a$ with~$b=2(1+T+n)+(c-2)$ and~$a=3+4Tn+(c-2)$.  
    Due to the filter requests, if~$m$ ascending subroutes are used, each ascending subroute must contain exactly~$c-1$ promise requests, as no two filter requests can be in the same subroute (due to~$\alpha<2$). 
    Thus, Observation~$3$, and consequently Observations~$4$ and~$5$, still apply as each of the~$m$ subroutes has only one seat not occupied by promise requests, just as in the case with capacity~$2$. The correctness is thus analogous to the case with capacity~$2$.
    For the inverse direction,~$(c-2)$ promise request are added to each constructed subroute. As the service promise was adapted, it is not violated if the promise requests are picked-up and dropped-off in the same order. Thus, the constructed route is feasible.

    We now generalize from arbitrary capacities to arbitrary number~$k$ of vehicles. 
    We modify the instance by adding~$(k-1)mc$ additional promise requests. As before,
    $\optturns = 2m-1$ for this \textsc{MinTurn} instance if and only if $S$ has a $3$-partition.
    Indeed, if~$S$ has a~$3$-partition, we can create one route, as in the case for~$k=1$, that serves all of the value, plug and filler requests, as well as~$c-1$ promise requests.
    It then remains to serve the~$(k-1)mc$ added requests. We know that each ascending subroute can serve~$c$ of these requests. Furthermore, each route with~$2m-1$ turns can contain~$m$ ascending subroutes. Thus, the remaining~$k-1$ routes can serve all the remaining promise requests. 
    Conversely, if $\optturns = 2m-1$ for this \textsc{MinTurn} instance, the corresponding collection~$R$ contains at most~$km$ ascending subroutes.
    Since all~$m$ filter requests have to be in separate subroutes, each of the remaining~$(k-1)m$ subroutes must contain~$c$ promise requests while the subroutes with a filter request contain~$c-1$ promise requests. As the $(k-1)m$ subroutes with~$c$ promise requests cannot contain any other requests, the remaining~$m$ subroutes can be rearranged into one route, see~\cref{lem:route-joining}, and thus this corresponds to a~$3$-partition of~$S$.

    Constructing these generalized instances still takes pseudo-polynomial time. As \textsc{3-Partition} is strongly \NP-hard, it follows that the~\textsc{MinTurn} problem is strongly \NP-hard.
    \qed
\end{proof}

\SPSCHardness*
\label{thm:sp-sc-nphard*}
\begin{proof}
    The construction is described in the main part of the paper. We now show that~$(S,m,T)$ is a \textsc{yes}-instance of \textsc{3-Partition} if and only if~$\optturns = 2m-1$ for the constructed \textsc{MinTurn} instance, before generalizing the reduction to higher values of~$k$ and~$c$.
   
    From~\cref{thm:poly-lidarp}, we know that an optimal solution to the constructed \textsc{liDARP} instance serves all requests. 
    As we see later (in Observation 1), each filter request has to be served by a different subroute, thus requiring at least~$m$ ascending subroutes to serve all requests. Thus, we need at least~$2m-1$ turns to serve all requests.

    Assume that we have an optimal solution to the \textsc{liDARP} that uses~$2m-1$ turns. 
    As we have only one vehicle and only ascending requests, we can thus assume that we have~$m$ ascending subroutes.
    
    \emph{Observation~$1$:} Each ascending subroute serves exactly one filter and one promise request. Indeed, as the service promise~$\alpha$ is less than~$2$ the maximum ride time of a filter request is less than~$2$, as the direct time distance is~$1$ due to the shortcuts. Thus, serving another filter request would require a detour of at least~$1$, thereby violating the service promise. Thus, each of the~$m$ ascending subroutes must also contain exactly one promise request as these subroutes serve all requests and filter and promise requests pairwise overlap.

    \emph{Observation~$2$:} The service promise allows a maximum delay of~$b=2T$ for each promise request, as, by using the provided shortcuts, the direct time distance between~$h_\textrm{s}$ and~$h_\textrm{e}$ is~$a=2+2n+2T$.

    \emph{Observation~$3$:} Serving a value request~$p_i$ requires a detour of~$2s_i$. This is because the vehicle has to drive along the sequence~$h_i^1,h_i^2,h_i^3,h_i^4$
    adding~$2s_i$ distance compared to using the shortcut between~$h_i^1$ and~$h_i^4$.

    \emph{Observation~$4$:} Combining observations~$2$ and~$3$, we conclude that the sum of the values corresponding to the value requests served by a subroute is~$T$. Let~$I_j$ be the set of indices~$i$ of value requests~$p_i$ served by subroute~$r_j$. Then~$\sum_{i\in I_j}s_i=T$. Setting~$S_j:=\set{s_i\mid i\in I_j}$ yields a valid~$3$-partition of~$S$.

    Conversely, if there exists a $3$-partition~$S_1,\dots,S_m$ of~$S$, we create~$m$ ascending subroutes~$r_1,\dots,r_m$ and assign for each~$j\in[m]$ the value requests~$p_i$ corresponding to~$s_i\in S_j$ to the subroute~$r_j$. We further assign one filter and promise request to each ascending subroute. Then, the delay in the ride time of a promise request is~$2T$ due to serving the value requests, which is the delay permitted by the service promise. Thus, each subroute is feasible. Joining these subroutes into a single route, we obtain a feasible route, see~\cref{lem:route-joining}, that serves all requests and has~$2m-1$ turns.

    For higher capacities, we use the same idea as in the reduction from~\cref{thm:sp-st-nphard}. We add more promise requests such that we have~$m(c-1)$ in total. If~$m$ ascending subroutes are used, each ascending subroute must then contain~$(c-1)$ promise requests. As this leaves a free capacity of~$1$ as well as a maximum delay of~$2T$ for the remaining drive, the correctness is analogous to the case of capacity~$2$.

    For a higher number of vehicles, the generalization is analogous to the generalization in the proof of~\cref{thm:sp-st-nphard}, adding~$(k-1)mc$ more promise requests.

    Differing from the reduction in~\cref{thm:sp-st-nphard}, we can construct this instance in polynomial time, as our number of stops does not depend on~$T$. It follows, that the \textsc{MinTurn} problem is strongly \NP-hard without time windows, without service times, as well as without turn times.
    \qed
\end{proof}

\InApproximable*
\label{cor:mt-approx-nphard*}
\begin{proof}
    We show this for the first case by reducing \textsc{3-Partition} to the \textsc{MinTurn} problem. Let~$(S,m,T)$ be an instance of \textsc{3-Partition}. We construct a \textsc{MinTurn} instance as in the proof of~\cref{thm:sp-st-nphard} for~$k=1$. We then change the number of vehicles in this instance to~$m$. As we have shown earlier,~$S$ has a~$3$-partition if and only if there is an optimal \textsc{liDARP} solution that uses~$m$ ascending subroutes. Since we now have~$m$ vehicles, this corresponds to each vehicle having exactly one ascending subroute, and thus one turn. If~$S$ does not have a~$3$-partition, more than~$m$ ascending subroutes must be used. By the pigeonhole principle, this means that there is at least one vehicle that has two ascending subroutes, and thus at least~$3$ turns. Therefore, if~$S$ has a~$3$-partition, then~$\optturns=1$ for the constructed \textsc{MinTurn} instance, otherwise~$\optturns\geq 3$.
    As this instance can be constructed in pseudo-polynomial time, it follows that it is strongly \NP-hard to approximate the \textsc{MinTurn} problem with a factor better than~$3$.

    The proof for the second case is analogous, using the construction from~\cref{thm:sp-sc-nphard} instead. 
    \qed
\end{proof}

\MTTWHardness*
\label{thm:mt-tw-nphard*}
\begin{proof}
    We  first show the reduction for~$k=1$ and~$c=1$ before we extend it to higher values of~$k$ and~$c$. Given an instance~$(S,m,T)$ of \textsc{3-Partition} with~${S=\set{s_1,\dots,s_n}}$, we construct the following instance of \textsc{liDARP}:  
    we have $s^{\max}:=\max_{s_i\in S} s_i+1$ stops~$H=\angles{0,\dots,s^{\max}-1}$. 
    The distance between subsequent stops is~$1$ and we do not have shortcuts.
    We have two types of requests: for each value~$s_i\in S$ we create a \emph{value request}~$p_i=([0,s_i],[0,2mT+2m-1])\in P_{\textrm{V}}$. 
    We also add~$m$ \emph{separator requests}~$P_\textrm{sep}=\set{p_\textrm{sep}^j\mid j\in[m]}$ with the separator requests being~$p_\textrm{sep}^j=([0,1],[2jT+2(j-1),2jT+2j-1])$.
    The turn time and service time are set to $0$, and we do not have a service promise.

    \begin{figure}
        \centering
        \includegraphics{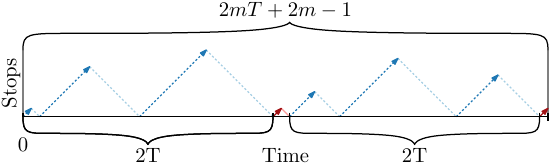}
        \caption{A tour serving all requests of an instance constructed in the reduction from \textsc{3-Partition} to the \textsc{MinTurn} problem with~$m=2$. The value requests in~$P_{\textrm{V}}$, represented by (blue) dashed arrows, are served in the intervals between the separator requests, represented by (red) solid arrows. The subtours used to return from serving a request are represented by lighter lines. }
        \label{fig:mt-tw-nphard}
    \end{figure}

    Since all requests originate at stop~$0$ and the capacity is~$1$, we need exactly two subtours per served request, except for the last one, as the vehicle does not need to return to~$0$. 
    Therefore, the number of turns required to serve all requests with a single tour is~$2m+2n-1$.
    Thus, it is~$\optturns=2m+2n-1$ if and only if there is a feasible tour that serves all requests. We will now show that such a tour exists if and only if there is a~$3$-partition of~$S$.

    Assume that there is a feasible tour that serves all requests.
    For the separator requests in~$P_\textrm{sep}$, the time between earliest pick-up and latest drop-off corresponds to the direct time distance of these requests. This means that if a separator request~$p_j\in P_\textrm{sep}$ is served, it has to be picked-up precisely at time ${2jT+2(j-1)}$ and dropped-off at time~$2jT+2j-1$. 
    If all of them are served, the value requests are relegated to be served between the fixed times where the separator requests are served. Thus, the value requests have to be served in the time intervals $[2T(j-1)+2(j-1),2jT+2(j-1)]$ for~$j\in[m]$. As each of these~$m$ intervals has size~$2T$ and the time needed to serve all of the value requests and return to station $0$ is $2mT$, if all value requests are served, the vehicles need to be busy all of the time. Thus, the intervals must be completely filled by the value requests. For the set of indices~$I_j$ of value requests served in the $j$-th interval, it thus holds that~$\sum_{i\in I_j}s_i=T$. Setting~$S_j=\set{s_i\mid i\in I_j}$ for~$j\in[m]$ we obtain a valid $3$-partition of~$S$.

    Conversely, if there is a $3$-partition~$S_1,\dots,S_m$ of~$S$, we group the value requests into sets~$P_j=\set{p_i\mid s_i\in S_j}$ for~$j\in[m]$. 
    We then construct a~tour as follows: all separator requests are served at their precise times. For~$j\in[m]$, the value requests in~$P_j$ are served in the~$j$-th time interval. As the direct time distances of the requests in~$P_j$ sum to~$T$, the resulting tour respects the time constraints and is feasible.

    To show that the result holds also for higher capacities $c>1$, we modify the proof as follows: in our reduction, we
    duplicate all requests, that is, we have~$c$ copies~$p_i^\ell$ for $\ell\in [c]$ of each value request created for an element~$s_i\in S$ and~$c$ copies of of each separator request.
    If the vehicle is always full when leaving stop~$0$,~$2m+2n-1$ turns are needed to serve all requests. 
    If all requests are served, the separator requests must be still served at their precise times, with the value requests being delegated to the~$m$ time intervals.
    We now show that for value requests, if all requests can be served, there is a tour such that all~$c$ requests created for an element~$s_i$ are served together. We prove this by incrementally creating such a tour. Assume we have a tour that serves all requests, but not all copies of a value request~$p_i$ are served together.
    Let~$s_i$ be the largest number for which not all copies of the corresponding value request are transported together, and let~$p_i^\ell$ be the request with lowest index that is not transported together with~$p_i^1$. 
    The request~$p_i^1$ must be transported together with another, shorter value request~$p'$. 
    The lengths of the two subtours (and its return subtours) that~$p_i^\ell$ and~$p'$ are assigned to, would not increase if we swapped the two requests~$p_i^\ell$ and~$p'$ (the time windows also allow the swapping, as they are the same for all value requests).
    By applying this process iteratively, we obtain a feasible tour where each value request is transported together with its copies.
    We conclude that if there is a \textsc{liDARP} solution that serves all requests, there is a \textsc{liDARP} solution that has~$2m+2n-1$ turns, and transports all value requests created for an element~$s_i$ together. It is easy to see that also all copies of a the separator requests are transported together, as they have precise times.
    Thus, as in the case of~${c=1}$, we can partition~$S$ according to the intervals in which the corresponding value request are served.
    Conversely, if~$S$ has a $3$-partition, we can proceed analogously to the case of~$c=1$ to construct a feasible tour serving all request with~$2m+2n-1$ turns, merely taking care that all duplicates of a request are served together.
    
    For higher number~$k>1$ of vehicles, we create separate \emph{service areas} for the vehicles within which we repeat the construction for~$k=1$. 
    That is, we consider~$k\cdot s^{\max}$ stops, where after a sequence of~$s^{\max}$ stops with distance~$1$ we have a longer distance~$2mT+2m$, that would not allow to serve a request after crossing it as all time windows have ended.
    We duplicate the requests into~$k$ collections of requests~$P'={\bigcup_{\ell\in [k]} P_\ell}$, with the requests in~$P_\ell$ being translated by~$(\ell-1)\cdot s^{\max}$ stops, i.e., all requests in the set~$P_\ell$ originate at stop~$(\ell-1)\cdot s^{\max}$.
    As all requests have to be served in the time interval~$[0,2mT+2m-1]$, each vehicle can only serve requests from one collection~$P_\ell$. 

    As the instance(s) can be constructed in polynomial time, and \textsc{3-Partition} is strongly \NP-hard, it follows that both \textsc{liDARP} and \textsc{MinTurn} problem are strongly \NP-hard. 
    \qed
\end{proof}

\subsection*{Completion of Proofs for the \emph{Parametrized Algorithms}}

In order to prove the existence of an \FPT-algorithm as claimed by Theorem~\ref{thm:minturn-fpt}, we need to prove a number of auxiliary results, which we do in the following. 

\begin{lemma}
    \label{lem:bound-served-pax}
    A single vehicle with capacity~$c$ can serve at most~$t\cdot c$ requests in time~$t$.
\end{lemma}
\begin{proof}
    As the origin and destination of a request need to be distinct, and we assume integer time distances, each request has to be transported a time distance of at least~$1$. Therefore, even if the service time is~$0$, the number of requests that can be served by a vehicle in a time interval of~$1$ is limited by its capacity~$c$. Therefore, in~$t$ time, at most~$t\cdot c$ requests can be served by a single vehicle.
    \qed
\end{proof}

\begin{lemma}
    \label{lem:event-based-xp}
    The \textsc{MinTurn} and \textsc{liDARP} problem can be solved in~$O^*(n^{2\cdot t\cdot c\cdot k})$ time.
\end{lemma}
\begin{proof}
    The algorithm is based on the \emph{event-based graph}, a concept introduced by Gaul~et~al.~\cite{DBLP:journals/eor/GaulKS22} and adapted to the \textsc{liDARP} in~\cite{reiter_line-based_2024}. In this graph, each vertex corresponds to an event, a pick-up or drop-off, together with the \emph{occupancy} of the vehicle (the passengers in the vehicle) at that point in time. Thus the graph contains~$O(n^c)$ vertices. Note that each vertex corresponds to a waypoint. If the waypoints corresponding to two vertices can appear consecutively in a feasible route, there is a directed edge between them. There is further a \emph{start vertex} connected to all pick-up events where there are no other passengers in the vehicle yet as well as an \emph{end vertex} connected to all drop-off events where the vehicle is empty afterwards.

    A route thus corresponds to a simple path in this graph from the start to the end vertex. As we know from~\cref{lem:bound-served-pax} that each vehicle can serve at most~${t\cdot c}$ requests, it follows that paths corresponding to feasible routes have lengths at most~$2\cdot t\cdot c$. Thus, there are~$O^*(n^{2\cdot t\cdot c})$ feasible routes.

    After enumerating all feasible routes, which we can do in~$O^*(n^{2\cdot t\cdot c})$ as testing the feasibility of a route is polynomial, see~\cref{lem:route-to-tour}, we check for each collection of up to~$k$~routes whether it is feasible, i.e., whether no two routes serve the same request. We then determine the optimal collections regarding the number of served requests. For the \textsc{liDARP}, we simply return one of these collections. For \textsc{MinTurn}, we determine~$\optturns$ by iterating over all collections. 
    As checking feasibility and calculating~$\max_{r\in R}\abs{r}$ for a collection~$R$ is possibly in polynomial time, the runtime is dominated by the number of collections, which is~$O^*(n^{2\cdot t\cdot c\cdot k})$.
    \qed
\end{proof}

To obtain an \FPT-runtime, we slightly modify the algorithm such that it reduces inputs to a predetermined size. For this we need the following lemma.

\begin{lemma}
    \label{lem:bound-requests}
    There are at most~$h^2\cdot t^2$ distinct requests.
\end{lemma}
\begin{proof}
    Each request is characterized by its origin, destination, earliest pick-up, and latest drop-off. For the origin and destination there are~$h$ options each, while for the earliest pick-up and latest drop-off there are up to~$t$ possibilities each. Thus, there are at most~$h^2\cdot t^2$ distinct requests.
    \qed
\end{proof}

\MinTurnFPT*
\label{thm:minturn-fpt*}
\begin{proof}
Since at most~$t\cdot c\cdot k$ requests can be served by~$k$ vehicles, see~\cref{lem:bound-served-pax}, a distinct request occurring more than~$t\cdot c\cdot k$ times in the input can be reduced to~$t\cdot c\cdot k$ occurrences without changing the solution of the problem. Thus, the reduced input has length at most~$h^2\cdot t^2\cdot (t \cdot c\cdot k)$. Using the algorithm from~\cref{lem:event-based-xp} results in a runtime of~$O^*((h^2\cdot t^3\cdot c\cdot k)^{2\cdot t\cdot c\cdot k})$.
\end{proof}

In order to prove \cref{thm:minturn-xp-no-tw}, we use the following lemma.
\begin{lemma}
    \label{lem:subroute-bound}
    A feasible subroute can serve at most~$c\cdot h$ requests and thus contains at most~$2\cdot c\cdot h$ waypoints. 
\end{lemma}
\begin{proof}
    The number of requests that can be transported together between neighboring stops is limited by the capacity of the vehicle. As each passenger has to be transported at least one stop along the line, a subroute can serve at most~$c\cdot h$ requests. Thus, a subroute can contain at most~$2\cdot c\cdot h$ waypoints, as each served request contributes two waypoints, the pick-up and drop-off. 
    \qed
\end{proof}

\MinTurnXPNoTW*
\label{thm:minturn-xp-no-tw*}
\begin{proof}
    We say that two subroutes~$r=\angles{w_1,\dots,w_m}$ and~$r'=\angles{w_1',\dots,w_m'}$ are \emph{identical} if for each~$j\in [m]$ the requests picked-up or dropped-off in waypoints~$w_j$ and~$w_j'$ are copies of each other.

    As without time windows there are at most~$h^2$ distinct requests, compare \cref{lem:bound-requests}, and each subroute consists of at most~$2\cdot c\cdot h$ requests that can be served by a subroute, see~\cref{lem:subroute-bound}, and a subroute is defined by its waypoints, there are~$O^*(h^{4\cdot c\cdot h})$ distinct subroutes. As the feasibility of a subroute can be checked in polynomial time, see~\cref{lem:route-to-tour}, all distinct feasible subroutes can be obtained in time~$O^*(h^{4\cdot c\cdot h})$.
    
    According to~\cref{lem:mt-subroute-combination}, it suffices to determine the minimum number of feasible ascending (descending) subroutes that serve all ascending (descending) requests, as~$\optturns$ can be calculated directly from this.

    Determining the minimum number of feasible subroutes needed to serve all ascending (descending) requests is the \textsc{Multiset Multicover} problem, where the up to~$h^2$ distinct requests have to be served up to~$n$ times (due to duplicates) and the up to~$h^{4\cdot c\cdot h}$ multisets that cover the requests (the subroutes) may each serve a distinct request multiple times (as they can serve multiple copies of a request). As shown by Hua~et~al.~\cite{DBLP:journals/tcs/HuaWYL10}, the optimal objective value of this problem can be found in time~$O^*(n^{h^2}\cdot h^{4\cdot c\cdot h})$. 
    \qed
\end{proof}

\end{document}